\documentclass[a4paper,UKenglish,cleveref, autoref, thm-restate]{lipics-v2021}
\usepackage{csquotes}
\usepackage{todonotes}
\usepackage[svgnames]{xcolor}

\usepackage[normalem]{ulem}

\newcommand{\set}[1]{\{ #1 \}}

\newcommand{\lpt}{\mathsf{lpt}}
\newcommand{\lct}{\mathsf{lct}}
\newcommand{\tw}{\mathsf{tw}}
\newcounter{tbox}

\usepackage{tcolorbox}

\newtheorem{question}[theorem]{Question}

\nolinenumbers

\tcbuselibrary{skins}

\newtcolorbox{mybox}[1]{minipage boxed title*=-2cm,
enhanced,attach boxed title to top center=
{yshift=-3mm,yshifttext=-1mm},colback=Lavender!30!white,
boxed title style={size=small,colback=Lavender},coltitle=black,
center title,title={#1}}

\title{The Complexity of Computing a Gallai Vertex} 
\title{The Gallai Vertex Problem is $\Theta_2^p$-Complete}

\funding{This work was supported by the VILLUM Foundation grant (VIL37507) ``Efficient Recomputations for Changeful Problems'', and the Independent Research Fund Denmark (DFF), grant agreement number 2098-00012B.}

\author{Amir Nikabadi}{ IT University of Copenhagen, Denmark }{amir@itu.dk}{0009-0002-1446-1935}{%
}

\author{Eva Rotenberg}{ IT University of Copenhagen, Denmark }{erot@itu.dk}{0000-0001-5853-7909}{}

\author{Lasse Wulf}{ IT University of Copenhagen, Denmark }{lasw@itu.dk}{0000-0001-7139-4092}{}

\authorrunning{A. Nikabadi, E. Rotenberg, L. Wulf} 

\Copyright{A. Nikabadi, E. Rotenberg, L. Wulf} 


\category{} 

\relatedversion{} 

\acknowledgements{The authors want to thank Jannik Matuschke for bringing our attention to the class $\Theta^p_2$ during an unrelated discussion.}


\EventEditors{}
\EventLongTitle{}
\EventShortTitle{}
\EventAcronym{}
\EventYear{}
\EventDate{}
\EventLocation{}
\EventLogo{}
\SeriesVolume{}
\ArticleNo{}

\begin{document}
\thispagestyle{empty}

\maketitle
   
\begin{abstract}
When a graph $G$ admits a vertex $v$ that is contained in all its longest paths, we call $v$ a \emph{Gallai vertex}. These are named after Gallai, who in 1966 asked the question if it is true that every connected graph contains such a vertex. This was soon answered in the negative by Walther and Zamfirescu, who presented a graph in which every vertex is omitted by some longest path of the graph.

In spite of its long history, the Gallai Vertex Problem, i.e. determining whether a graph has a Gallai vertex, was, until now, neither known to be NP- nor coNP-hard. In this work, we show something much stronger, as
we completely settle the computational complexity of determining whether a graph has a Gallai vertex:
we show that it is complete for the complexity class $\Theta_2^p = \text{P}^{\text{NP}[\log n]}$.
This class, also known as \enquote{parallel access to NP}, is a complexity class larger than NP situated just below the class $\Sigma^p_2$ in Stockmeyer's polynomial hierarchy.

In more generality,
the \textit{longest path transversal number} of a connected graph is the minimum size of a set of vertices that intersects all its longest paths. I.e. if the graph has a Gallai vertex, its longest path transversal number is $1$. 
Thus, as a consequence of our theorem, the longest path transversal number of a graph cannot be approximated in polynomial time by a factor better than 2, unless $\text{P} = \text{NP}$. 
In fact, using related techniques, we show a strengthening of this result: For any constant $C$, if there is a graph with longest path transversal number $C$, then there is no polynomial time algorithm for approximating the longest path transversal number by a factor better than $C$, unless $\text{P} = \text{NP}$. In particular, this excludes approximation by a factor below $3$. Similar results hold for the longest cycle transversal.
\end{abstract}

\newpage
\setcounter{page}{1}

\section{Introduction}
\label{sec:introduction}

\begin{figure}[t]
    \centering
    \includegraphics[width=0.2\textwidth,page=1]{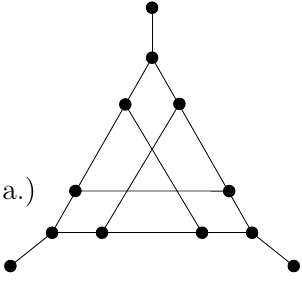}
    \hfill
    \includegraphics[width=0.2\textwidth,page=2]{img/walther-zamfirescu.pdf}
    \hfill
    \includegraphics[width=0.2\textwidth,page=3]{img/walther-zamfirescu.pdf}
    \caption{a.) The Walther-Zamfirescu graph. b.) The longest path has length 9. The paths $P_1$ and $P_2$ have length 9 and every vertex is omitted by some rotation/reflection of $P_1$ or $P_2$.}
    \label{fig:petersonblow}
\end{figure}

It is a common exercise in an introductory graph theory course to show that every two longest paths in a connected graph share a common vertex. An old question of Gallai~\cite{Gallai68} asks whether all longest paths in a connected graph share a common vertex. Such a vertex is called a \emph{Gallai vertex}. According to a counterexample due to Walther~\cite{walther1969nichtexistenz} and Zamfirescu~\cite{zamfirescu1976longest}, there exists a graph (depicted in~\Cref{fig:petersonblow}) such that every vertex is omitted by some longest path of the graph, providing a negative answer to Gallai's question. 
However, one could be modest and ask for a small subset $S$ of vertices, instead of only one vertex, such that every longest path in the graph has a vertex from the subset $S$. Such a subset is called a \emph{longest path transversal}. 
The \textit{longest path transversal number} of a connected graph $G$, denoted by $\lpt(G)$, is the minimum cardinality of a longest path transversal in $G$. 
 In this language, Gallai's question asked whether $\lpt(G)$=1 for all graphs $G$.

Note that the graph in \Cref{fig:petersonblow} has no Gallai vertex, in particular it has 
$\lpt = 2$. The smallest known graph $G$ with $\lpt(G)=3$ is given by Zamifirescu~\cite{zamfirescu1976longest} and has 270 vertices.
Amazingly, even 50 years after Gallai's original question it is still unknown whether there exists a connected graph $G$ with $\lpt(G)\geq4$. 
Accordingly, it has become a major open question, raised independently by Walther and Zamfirescu \cite{walther1969nichtexistenz,zamfirescu1976longest}, whether there exists a universal constant bounding the longest path transversal number of any connected graph. 
\begin{question}\label{ques:constantupper}
Is there a constant $c\in \mathbb{N}$ such that every connected graph $G$ satisfies $\lpt(G) \leq  c$?
\end{question}
While an extensive body of work has been devoted to the study of~\Cref{ques:constantupper}, the best known upper bounds are far from  constant. 
The first nontrivial upper bound was given by Rautenbach and Sereni \cite{rautenbach2014transversals}, who showed that for every $n$-vertex connected graphs $G$, it holds that $\lpt(G) \leq \lceil \frac n 4 - \frac{n^{2/3}}{90}\rceil$. 
This was improved to sublinear bounds, first to $\lpt(G) = O(n^{3/4})$ by  Long Jr., Milans, and Munaro
\cite{long2021sublinear}, then to $\lpt(G) = O(n^{2/3})$ by  Kierstead and Ren \cite{kierstead2023improved}. Norin, Steiner, Thomassé, and Wollan~\cite{norin2025small} proved the currently best known upper bound of $\lpt(G) = \mathcal{O}(\sqrt{n})$.

\Cref{ques:constantupper} remains widely open even for special classes of graphs, such as chordal graphs, though it has been resolved positively for some of its subclasses~\cite{balister2004longest, jobson2016detour}, for circular-arc graphs~\cite{balister2004longest,joos2015note}, and for bipartite permutation graphs~\cite{cerioli2020}. 
Several results also provide partial answers for~\Cref{ques:constantupper} when restricted to $H$-free graphs 
\cite{long2023non, de2025hitting, lima2025longest}.
For the special case where $G$ is chordal, Harvey and Payne~\cite{harvey2023intersecting} showed $\lpt(G) \leq 4\lceil(\omega(G))/5 \rceil$, while
Long Jr., Milan and Munaro~\cite{long2024longest} showed that if $G$ is an $n$-vertex connected chordal graph, then $\lpt(G) = \mathcal{O}(\log^{2}n)$.
One can obtain that $\lpt(G) \leq \tw(G)+1$ as a consequence of the result of Seymour and Thomas~\cite{seymour1993graph} on brambles.
Several authors have provided examples showing that many interesting graph classes contain graphs with $\lpt(G) \geq 2$ or $\lpt(G) \geq 3$. We refer the interested reader to the surveys \cite{zamfirescu2001intersecting,shabbir2013intersecting}.

Closely related to the longest path transversal is the longest cycle transversal number $\lct(G)$. Here, given a graph $G$, the question becomes to find a small set of vertices hitting all longest cycles of $G$.
Very similar bounds as for $\lpt(G)$ were shown for $\lct(G)$ in \cite{rautenbach2014transversals,long2021sublinear,kierstead2023improved,norin2025small}. 
Regarding the differences between longest path and longest cycle transversal, we refer the reader to \cite{shabbir2013intersecting,norin2025small}.

\subparagraph{Complexity.} 
We consider in this paper the computational complexity to determine if a given graph has a Gallai vertex. Formally, this is the following problem.

\begin{mybox}{\textsc{Gallai Vertex}}
\textbf{Input:} A connected graph $G$.\\
\textbf{Task:} Decide whether $G$ admits a Gallai vertex.
\end{mybox}

Even though Gallai's question is from 1966, little was known until now about the computational complexity of \textsc{Gallai Vertex}. One possible explanation for this lack of knowledge is that it seems very difficult to find graphs with $\lpt(G) > 1$, and therefore seems hard to reason about such graphs.
Unlike many similar-looking computational problems, \textsc{Gallai Vertex} was neither known to be NP-hard, nor was it easy to imagine any polynomial-sized polynomial-time checkable certificate for the problem. In particular, any such certificate would have to argue about all (potentially exponentially many) longest paths in the graph. And indeed, as a consequence of our theorem, $\textsc{Gallai Vertex} \not\in \text{NP}$ unless common complexity-theoretic assumptions fail.

\subparagraph{Our contribution.} 
We are the first to show that \textsc{Gallai Vertex} is NP-hard. 
Moreover, in this paper, we 
completely settle the complexity of the problem. We show that \textsc{Gallai Vertex} is complete for the complexity class $\Theta_2^p = \text{P}^{\text{NP}[\log n]}$. 
This complexity class, later named by Wagner~\cite{wagner1987more}, was first studied by Papadimitriou and Zachos~\cite{Papadimitriou1983TwoRO}, and is defined as the class of problems that can be solved with an algorithm that uses polynomial time and $O(\log n)$ oracle calls to NP. This class has seen attention also from the research areas of logic and computational social choice. 
Another contribution of this paper is to present a new and arguably very natural $\Theta_2^p$-complete problem. 

\nocite{papadimitriou2003computational,hemaspaandra1997exact,hemaspaandra2005complexity,riege2006completeness} 

\subparagraph{Inapproximability.}
Our result implies that it is $\Theta_2^p$-hard to distinguish between graphs with $\lpt(G) = 1$ and $\lpt(G) \geq 2$. Therefore, we obtain as an immediate consequence, that it is $\Theta_2^p$-hard (and therefore also NP-hard) to approximate $\lpt(G)$ to a factor better than 2. 

We strengthen this result to the following: For any constant $C$, if there is a graph $H$ with $\lpt(H) =C$, then it is $\Theta_2^p$-hard to distinguish whether a given graph $G$ has $\lpt(G) = 1$ or $\lpt(G) \geq C$. 
In particular, since there exists a graph of $\lpt(H) = 3$ shown in \cite{grunbaum1973vertices,zamfirescu1976longest}, this shows polynomial time inapproximability by a factor of $3-\varepsilon$ for any $\varepsilon>0$.

If $\lpt(G)$ can take arbitrarily large values, our results prove that no constant-factor approximation algorithm exists. Furthermore, this inapproximability holds even for planar graphs.

\subparagraph{Longest cycle.} Finally, we show using similar ideas that it is $\Theta^p_2$-hard to approximate the longest cycle transversal. Furthermore, it is $\Theta^p_2$-complete to decide if there exists a single vertex that hits all longest cycles. This holds even in the $2$-connected planar case.

\section{Preliminaries and further related work}
\label{sec:prelim}
We use $\mathbb{N}$ to denote the set of positive integers. We let $[n]:= \{1,\dots, n\}$ for every $n\in \mathbb{N}$.
Graphs in this paper, unless stated otherwise, have finite vertex sets and no loops or parallel edges. Let $G = (V,E)$ be a graph. 
For $X \subseteq V(G)$, we denote the subgraph of $G$ induced by $X$ as $G[X]$, that is $G[X] = (X, \{uv \colon u, v \in X \mbox{ and } uv \in E \})$. We let $G - v := G[V \setminus \{v\}]$. 
We let $P = p_1p_2\dots p_k$ denote a path in $G$.
The length of $P$ is the number of edges in $P$.
We call the vertices $p_1$ and $p_k$ the endpoints of $P$ and say that $P$ is a path from $p_1$ to $p_k$. For some (possibly induced sub-) graph $G$, we denote the length of the longest path inside $G$ by $\lambda(G)$.
The \textsc{3SAT} problem is the following decision task: Given a boolean formula in 3-conjunctive normal form, decide if there exists a satisfying assignment. 

\subparagraph{The class $\Theta_2^p$.}
Papadimitriou and Zachos~\cite{Papadimitriou1983TwoRO} introduced the class $\Theta_2^p = \text{P}^{\text{NP}[\log n]}$ as the class of problems that can be computed in polynomial time using $O(\log n)$ many oracle calls to NP.
The class $\Theta^p_2$ falls in between the first two levels of Stockmeyer's \emph{polynomial hierarchy} \cite{stockmeyer1976polynomial}: We have $\text{NP} \subseteq \text{P}^{\text{NP}[\log n]} = \Theta_2^p \subseteq \text{P}^\text{NP} \subseteq \text{NP}^\text{NP}$.
Complexity theorists have the strong belief that all these inclusions are strict (the polynomial hierarchy does not collapse), but the true status is unknown.
Wagner~\cite{wagner1990bounded} introduced the name $\Theta_2^p$ and showed that the class can be characterised in many equivalent ways. Kadin~\cite{kadin1989pnp} has proven that if NP has a sparse Turing-complete set then the polynomial hierarchy collapses to $\Theta_2^p$.

The search for natural $\Theta_2^p$-complete problems has received attention from the community.
Prominently, in the area of computational social choice, the class $\Theta_2^p$ plays a role.
In 1876, Charles Lutwidge Dodgson, nowadays better known under his pen name Lewis Caroll,
proposed a voting rule for which it may be computationally nontrivial to determine who is the winner of the election~\cite{Dodgson1876Method}. Hemaspaandra, Hemaspaandra, and Rothe~\cite{hemaspaandra1997exact} proved that determining the winner under Dodgon's rule is $\Theta_2^p$-complete.
Similar results were soon after obtained for Young's rule and Kemeny's rule \cite{rothe2003exact,hemaspaandra2005complexity}. Further natural $\Theta_2^p$-complete results are known in AI and logic \cite{eiter1997complexity}.
We also mention \cite{riege2006completeness, matuschke2025stronger}.
Related completeness results also arise in the study of stability and criticality of graph parameters.
In their paper, Frei, Hemaspaandra, and Rothe~\cite{FreiHemaspaandraRothe2020} established $\Theta_2^p$-completeness of deciding whether $\xi(G-v)=\xi(G)$ for every vertex $v\in V(G)$, for each choice of $\xi$ as the independence number, clique number, or chromatic number.
Their results for the first two parameters follow from the corresponding vertex-cover criticality result of Burjons et al.~\cite{BurjonsEtAl2024}.
Specifically, Burjons et al.\ proved that deciding whether $\beta(G-v)<\beta(G)$ for every vertex $v\in V(G)$ is $\Theta_2^p$-complete, where $\beta(G)$ denotes the minimum vertex-cover size of $G$.

Hemachandra \cite{hemachandra1987strong}, and independently Köbler, Schöning,
and Wagner \cite{kobler1987difference} proved that the class $\Theta^p_2$ can be equivalently characterised as  $\Theta^p_2 = \text{P}^{\text{NP}[\log n]} = \text{P}*\parallel^\text{NP}$.
Here $\text{P}*\parallel^\text{NP}$ denotes the class of all problems that can be solved with so-called \emph{parallel access} to an NP-oracle (also sometimes called \emph{non-adaptive access}).
Formally, a decision problem $\Pi$ is contained in $\text{P}*\parallel^\text{NP}$ if there exists a polynomial-time Turing machine, which, upon receiving an instance $I_\Pi$ of problem $\Pi$ creates a table of polynomially many different membership queries to some NP-complete problem.
The Turing machine has to first write the entire table. It will then receive the answers of all the queries at once (i.e. in parallel).
It cannot make any additional queries and has to make its final decision about instance $I_\Pi$ based on the answers in the truth table.
As stated before, it is known that $O(\log n)$ many \emph{adaptive} NP-queries is equivalent to $\text{poly}(n)$ many \emph{non-adaptive} NP-queries~\cite{hemachandra1987strong,kobler1987difference}.

\section{Complexity of computing a Gallai vertex}

In this section, we prove that \textsc{Gallai Vertex} is $\Theta^p_2$-complete. We have to show containment and hardness for $\Theta^p_2$. 
We start with the containment. Then, before proceeding to $\Theta_2^p$ hardness, we present a proof of NP-and coNP hardness, based on a modification of the Walther-Zamfirescu graph in Figure~\ref{fig:petersonblow}. Finally, expanding these ideas even further, we show $\Theta_2^p$ hardness of the problem.

\subsection{\textsc{Gallai Vertex} is contained in $\Theta_2^p$.}

\begin{lemma}
\label{lem:containment}
    \textsc{Gallai Vertex} is contained in $\Theta^p_2$.
\end{lemma}
\begin{proof}
As explained in \cref{sec:prelim}, the class $\Theta^p_2$ can be characterised as the class of problems solvable in polynomial time while having access to polynomially many parallel NP-queries. Hence it is enough to show how \textsc{Gallai Vertex} can be solved with such a parallel oracle.
Consider the following decision problem \textsc{Path}: The input is a tuple $(G, k)$ of a graph $G$ and some integer $k \geq 1$. The question is if $G$ contains a path of length at least $k$. Clearly, $\textsc{Path} \in \text{NP}$. 
Consider now the problem $\textsc{Gallai Vertex}$. Given a graph $G$, we show to find out if $G$ has a Gallai vertex, making use of access to a parallel oracle for \textsc{Path}.
Let $n := |V(G)|$. The algorithm upon receiving $G$ creates a table consisting of a total of $(n+1)n$ queries as follows.
\begin{align*}
\bigcup_{k=1}^n &\{ \textsc{Path}(G, k)? \} \\
\cup \bigcup_{v \in V(G)}\bigcup_{k=1}^n &\{ \textsc{Path}(G - v, k)? \}
\end{align*}
Given the answer to these $O(n^2)$ queries, one can easily figure out in polynomial time the exact length of the longest path in the graph $G$, as well as the exact length of the longest path in $G - v$ for each $v \in V(G)$. In particular, we can find out if we can delete a vertex such that the length of the longest path decreases. 
Hence, by definition, we can decide if there exists a Gallai vertex. We conclude $\textsc{Gallai Vertex} \in \Theta^p_2$.
\end{proof}

\subsection{Warm up: \textsc{Gallai Vertex} is  NP- and coNP-hard.}

For the hardness part, as a stepping stone we first show the weaker statement that \textsc{Gallai Vertex} is both NP-hard and coNP-hard. We then show afterwards how to further strengthen the argument to show $\Theta^p_2$-hardness. 
We now begin with an explanation of the proof. Consider \cref{fig:WZ-paths}. It shows a modification of the Walther-Zamfirescu graph, which results from the original graph by appending three paths of length $k_T,k_L$ and $k_R$ each at the topmost ($T$), leftmost ($L$) or rightmost ($R$) vertex (in the following, we always assume $k_T,k_L,k_R \geq 1$). Let us call this graph $G(k_T, k_L, k_R)$. Does $G(k_T, k_L, k_R)$ have a Gallai vertex?
\begin{figure}[htpb]
    \centering
    \includegraphics[width=0.55\textwidth,page=1]{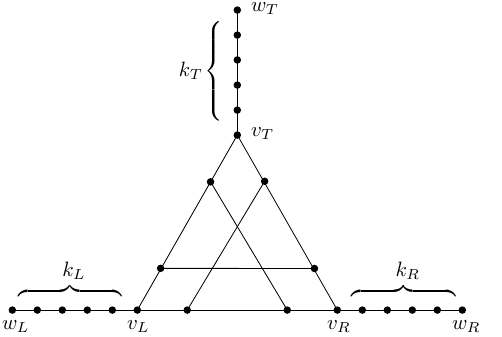}
    \caption{The graph $G(k_T,k_L,k_R)$ results from the  Walther-Zamfirescu graph.}
    \label{fig:WZ-paths}
\end{figure}

\begin{lemma}
\label{lem:3-paths}
        If $k_T = k_L = k_R$ then the graph $G(k_T,k_L,k_R)$ does not have a Gallai vertex. In all other cases, it does have one.
\end{lemma}
\begin{proof}
    First, assume that $k_T = k_L = k_R$. Then, it is easy to verify that the longest path in $G(k_T,k_L,k_R)$ has length $2k_T + 7$. 
    Furthermore, we can consider paths very similar to $P_1, P_2$ in \cref{fig:petersonblow}, which have length $2k_T + 7$, and for every vertex $v \in V(G(k_T,k_L,k_R))$ one of the rotations or reflections of the paths omits the vertex $v$. This shows that $G(k_T,k_L,k_R)$ does not have a Gallai vertex.

    Second, assume that not all of $k_T,k_L,k_R$ are equal. Due to rotational and reflectional symmetry we can w.l.o.g. assume $k_T \geq k_L \geq k_R$. Again, it is easy to verify that the longest path now has length $k_T + k_L + 7$. We claim that the vertex $v_T$ is a Gallai vertex. 
    Indeed, let $P$ be a path that does not visit $v_T$. 
    Then, $P$ is a path in the graph $G(k_T,k_L,k_R) - v_T$. Therefore it has length at most $\max\{k_T - 1, k_L + k_R + 7 \} < k_T +k_L + 7$. The strict inequality is due to our assumption about $k_T,k_L,k_R$. Hence $P$ is not a longest path.
    We conclude that $v_T$ is a Gallai vertex, as desired.
\end{proof}

The next step in the proof is to consider the same basic idea, but instead of three paths of length $k_T,k_L,k_R$ consider more complicated \emph{gadget graphs} $G(\varphi)$. The following lemma states which properties we expect from such gadget graphs.

\begin{lemma}
    \label{lem:gadget-ham-path}
    Given a \textsc{3SAT} formula $\varphi$, one can compute in polynomial time a graph $G(\varphi)$, two vertices $v,w$, and an integer $k \geq 1 $ such that
    \begin{itemize}
        \item If $\varphi$ is satisfiable, then the longest path in $G(\varphi)$ has length $k$.
        \item If $\varphi$ is not satisfiable, then the longest path in $G(\varphi)$ has length $k-1$.
    \end{itemize}
    Furthermore, in both cases, every longest path in $G(\varphi)$ starts at $v$ and ends at $w$.
\end{lemma}
\begin{proof}
\begin{figure}[htpb]
    \centering
    \includegraphics[width=0.5\textwidth]{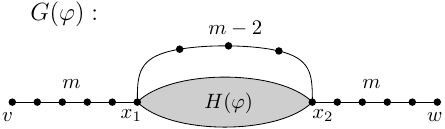}
    \caption{The gadget graph $G(\varphi).$}
    \label{fig:gadget-ham-path}
\end{figure}
    By the classic result of Karp \cite{DBLP:conf/coco/Karp72}, the Hamiltonian path problem is NP-complete. The proof of Karp implies the slightly stronger statement: Given the \textsc{3SAT} formula $\varphi$, one can compute in polynomial time a graph $H(\varphi)$ and two of its vertices, $x_1,x_2$, such that if $\varphi$ is satisfiable, 
    then there is a Hamiltonian path in $H(\varphi)$ from $x_1$ to $x_2$, and if $\varphi$ is not satisfiable there is no Hamiltonian path at all in $H(\varphi)$.
    Let $m := |V(H(\varphi))|$ be the number of vertices of the constructed instance.
    Consider now \cref{fig:gadget-ham-path}. We start with the graph $H(\varphi)$ and append two paths of length $m$ each to $x_1$ and $x_2$. We let $v, w$ denote the other endpoints of these paths. Furthermore, we add a path of length $m-2$ between $x_1$ and $x_2$. We let $k :=  3m-1$. This completes the description of the gadget $G(\varphi)$. We claim that $G(\varphi)$ has the desired properties.
    
    First, assume that $\varphi$ is satisfiable. Then clearly $G(\varphi)$ has a path of length $3m-1$ utilizing the Hamiltonian path of length $m-1$ in $H(\varphi)$. There cannot be a longer path, since for any path, the subpath inside $H(\varphi)$ can have length at most $m-1$. Hence $\lambda(G(\varphi)) = k$.
    
    Second, assume that $\varphi$ is not satisfiable. Then, every path fully contained in $H(\varphi)$ has length at most $m-2$. There exists a path of length $3m-2 = k-1$. It is possible to confirm with an easy case distinction that no longer path exists, i.e. $\lambda(G(\varphi)) = k-1$.

    Finally, consider a path $P$ of length $k$ in the yes-case, or of length $k-1$ in the no-case. It is again easy to see that $P$ must have one starting vertex at $v$ and the other at $w$.
 \end{proof}

\begin{lemma}\label{lem:npconp}
    \textsc{Gallai Vertex} is NP-hard and coNP-hard.
\end{lemma}
\begin{proof}

We use the Walther--Zamfirescu graph with attachment vertices $T,L,R$, as in Lemma~\ref{lem:3-paths}. We first claim the following.

\begin{claim}
\label{claim:gadgets-instead-of-paths}
Let $H_T,H_L,H_R$ be connected graphs with distinguished vertices $v_X,w_X \in V(H_X)$ for each label $X \in \{T,L,R\}$, and suppose that every longest path of $H_X$ has endpoints $v_X$ and $w_X$. For each $X \in \{T,L,R\}$, identify the vertex $v_X$ with the corresponding attachment vertex $X$ of the Walther--Zamfirescu graph, and let $k_X := \lambda(H_X)$.
Then the resulting graph has no Gallai vertex if $k_T = k_L = k_R$. If the values $k_T,k_L,k_R$ are not all equal, then the resulting graph has a Gallai vertex.
\end{claim}

\begin{proof}[Proof of claim]
Since $H_X$ meets the rest of the graph only in the vertex $v_X = X$, a simple path in the full graph can enter $H_X$ at most once. Hence its intersection with $H_X$ is a path in $H_X$, and therefore uses at most $k_X = \lambda(H_X)$ edges. Equality is possible only if this subpath is a longest path of $H_X$, in which case its endpoints are $v_X$ and $w_X$ by assumption. Thus, for the purpose of determining longest paths in the full graph, each $H_X$ behaves exactly like a pendant path of length $k_X$ attached at $X$. The conclusion now follows from Lemma~\ref{lem:3-paths}.
\end{proof}

\begin{figure}[h]
    \centering
    \includegraphics[width=0.55\textwidth,page=2]{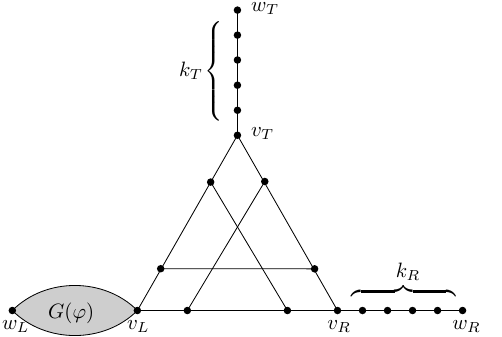}
    \caption{The instance used to show NP-hardness and coNP-hardness}
    \label{fig:NP-coNP-hardness}
\end{figure}

We first prove NP-hardness. Let $\varphi$ be an instance of \textsc{3SAT}. By \cref{lem:gadget-ham-path} we can compute in polynomial time a connected graph $G(\varphi)$ with distinguished vertices $v,w$ and an integer $k$ such that
\begin{itemize}
    \item if $\varphi$ is satisfiable, then $\lambda(G(\varphi)) = k$;
    \item if $\varphi$ is not satisfiable, then $\lambda(G(\varphi)) = k-1$;
    \item in both cases, every longest path of $G(\varphi)$ has endpoints $v$ and $w$.
\end{itemize}

We now use the construction shown in Figure~\ref{fig:NP-coNP-hardness}. More precisely, we start with the graph $G(k_T, k_L, k_R)$ from  \cref{fig:WZ-paths} and replace the left pendant path by the gadget $G(\varphi)$, identifying the vertex $v_L$ with the vertex $v$ from the gadget $G(\varphi)$. For NP-hardness, we choose both the top and right pendant paths to have length $k_T = k_R := k-1$.
If $\varphi$ is satisfiable, then the three attached pieces have longest-path lengths $(k,k-1,k-1)$, which are not all equal. By the claim, the resulting graph has a Gallai vertex.
If $\varphi$ is not satisfiable, then the three attached pieces have longest-path lengths $(k-1,k-1,k-1)$. By the claim, the resulting graph has no Gallai vertex.
Hence
$
\varphi \in \textsc{3SAT}$ if and only if
$G_{\mathrm{NP}}(\varphi) \in \textsc{Gallai Vertex}.$
This is a polynomial-time many-one reduction from \textsc{3SAT} to \textsc{Gallai Vertex}.

For coNP-hardness, we reduce from the coNP-complete \textsc{UNSAT} problem. In the \textsc{UNSAT} problem, the input is a formula $\varphi$ in 3-conjunctive normal form.
The question is whether $\varphi$ is unsatisfiable.
We again use the construction shown in Figure~\ref{fig:NP-coNP-hardness}, but now we choose the top and right pendant paths to have length $k_T = k_R := k$.
If $\varphi$ is satisfiable, then the three attached pieces have longest-path lengths $(k,k,k)$, so by the claim, the resulting graph has no Gallai vertex.
If $\varphi$ is not satisfiable, then the three attached pieces have longest-path lengths $(k-1,k,k)$, which are not all equal. By the claim, the resulting graph has a Gallai vertex.
Thus
$
\varphi \in \textsc{UNSAT}
\iff
G_{\mathrm{coNP}}(\varphi) \in \textsc{Gallai Vertex}.
$
This is a polynomial-time many-one reduction from \textsc{UNSAT} to \textsc{Gallai Vertex}. Since \textsc{UNSAT} is coNP-complete, it follows that \textsc{Gallai Vertex} is coNP-hard.
This finishes the proof of~\Cref{lem:npconp}.
\end{proof}

\subsection{\textsc{Gallai Vertex} is $\Theta_2^p$-hard}

We now show how to extend the argument to $\Theta^p_2$-hardness. In order to do this, we make use of a general technique by Wagner to prove $\Theta^p_2$-hardness, which is nowadays standard \cite{wagner1987more}.  Wagner's technique can be explained as follows. Consider the problem \textsc{SAT parity}.

\begin{mybox}{\textsc{SAT Parity}}
\textbf{Input:} An even-length sequence $(\varphi_1,\dots,\varphi_{2n})$ of $\textsc{3SAT}$ formulas. We have the promise that there exists an index $s \in \set{0,\dots,2n}$ such that $\varphi_1,\dots,\varphi_s$ are satisfiable, and $\varphi_{s+1},\dots,\varphi_{2n}$ are not satisfiable.\\
\textbf{Question:} Is the index $s$ an odd number?
\end{mybox}

We call the index $s$ the \emph{split index}. Wagner proved: For some decision problem $\Pi$ (here, $\Pi = \textsc{Gallai Vertex}$), if there is a polynomial-time computable function $f$ that maps arbitrary instances $x$ of \textsc{SAT Parity} to instances $f(x)$ of $\Pi$ such that $x$ is a yes-instance if and only if $f(x)$ is a yes-instance, then $\Pi$ is $\Theta^p_2$-hard. 
Informally, this can be understood as a reduction from $\textsc{SAT Parity}$ to $\Pi$.
However, the reason this is not a reduction in the usual sense is that instances of the problem \textsc{SAT Parity} contain a promise, but it cannot be tested in polynomial time (unless P=NP) whether this promise holds. 
Hence \textsc{SAT Parity} is not a computational decision problem in the usual sense. (See \cite[Thm. 5.2]{wagner1987more}, or \cite{riege2006completeness} for a formally very precise statement.)

In the context of our paper, this means that we need to start with an instance $(\varphi_1,\dots,\varphi_{2n})$ of $\textsc{SAT parity}$, and transform it into a graph $G$, 
such that $G$ has a Gallai vertex if the split index $s$ is odd, and $G$ does not have a Gallai vertex if the split index is even. How can we achieve such a feat? The main idea is to introduce so-called \emph{odd/even-gadgets} and show that they behave differently depending on whether the split index $s$ is even or odd. 
We remark that the value of $\lambda(G_\text{odd}),\lambda(G_\text{even})$ depends on $s$, but the truth of the equality \enquote{$\lambda(G_\text{odd}) = \lambda(G_\text{even})$?} depends only on the parity of $s$. 

\begin{lemma}
\label{lem:gadget-odd-even}
    Given an instance $(\varphi_1,\dots,\varphi_{2n})$ of \textsc{SAT Parity}, let $s \in \set{0,\dots,2n}$ denote its split index. One can construct in polynomial time two gadgets $G_\text{odd}$ and $G_\text{even}$ together with vertices $v_o, w_o \in V(G_\text{odd})$ and $v_e,w_e \in V(G_\text{even})$ such that
    \begin{itemize}
    \item If $s$ is even, then $\lambda(G_\text{odd}) = \lambda(G_\text{even})$.
        \item If $s$ is odd, then $\lambda(G_\text{odd}) > \lambda(G_\text{even})$.
    \item In both cases, every longest path in $G_\text{odd}$ ($G_\text{even}$, respectively) starts and ends at vertices $v_o,w_o$ ($v_e, w_e$, respectively).
    \end{itemize}
\end{lemma}
\begin{proof}
    The main idea is to combine the gadget from \cref{lem:gadget-ham-path} multiple times with itself. 
    By the lemma, for each $i=1,\dots,2n$, we can compute a graph $G(\varphi_i)$ and some integer $k_i$ such that $\lambda(G(\varphi_i)) = k_i$, if $\varphi_i$ is satisfiable, and  $\lambda(G(\varphi_i)) = k_i - 1$ otherwise. 
    Let $k := \max_{i=1}^{2n}k_i$. 
    If some $k_i < k$, then we can simply extend the gadget $G(\varphi_i)$ with a path of length $k-k_i$, while keeping all its properties. 
    Hence we can w.l.o.g.\ assume that $k_i = k$ for all $i = 1,\dots, 2n$.
    \begin{figure}[htpb]
    \centering
    \includegraphics[width=0.9\textwidth]{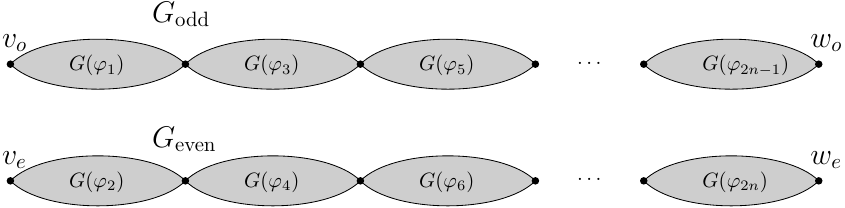}
    \caption{The gadgets $G_\text{odd}$ and $G_\text{even}$.}
    \label{fig:gadget-odd-even}
    \end{figure}
    Consider now \cref{fig:gadget-odd-even}. 
    The gadget $G_\text{odd}$ is defined as the concatenation of the odd-numbered gadgets $G(\varphi_1),G(\varphi_3),G(\varphi_5),\dots,G(\varphi_{2n-1})$ (by identifying vertex $w$ from the previous gadget with vertex $v$ from the next gadget). 
    The gadget $G_\text{even}$ is similarly defined as the concatenation of the even-numbered gadgets $G(\varphi_2),G(\varphi_4),G(\varphi_6),\dots,G(\varphi_{2n})$. 
    The vertices $v_o,w_o$ ($v_e, w_e$, respectively) are the leftmost/rightmost vertices of $G_\text{odd}$ ($G_\text{even}$, respectively). 
    It is clear from \cref{lem:gadget-ham-path} that every longest path in $G_\text{odd}$ has $v_o,w_o$ as its endpoints, and this is similarly true for $G_\text{even}$. What is the length of the longest path in $G_\text{odd}$ and $G_\text{even}$? 
    Note that this depends crucially on the split index $s$. This is because if formula $\varphi_i$ is satisfiable, the gadget $G(\varphi_i)$ contributes $k$ edges to the longest path, but if $\varphi_i$ is not satisfiable, it contributes only $k-1$. Hence the larger the split index $s$ is, the longer the longest path in $G_\text{odd}$ and $G_\text{even}$. 
    However, note that the precise value depends on whether $s$ is odd or even.
    \begin{itemize}
        \item If $s$ is even, then the $s$ satisfiable formulas $\varphi_1,\dots,\varphi_s$ are distributed evenly in $G_\text{odd}$ and $G_\text{even}$, hence
        \[
            \lambda(G_\text{odd}) = \lambda(G_\text{even}) =  \frac{s}{2}k + \frac{2n-s}{2}(k-1).
        \]
        \item However, if $s$ is odd, then out of the $s$ satisfiable formulas $\varphi_1,\dots,\varphi_s$, the gadget $G_\text{odd}$ contains one more than the gadget $G_\text{even}$, hence
        \[
        \lambda(G_\text{odd}) = \frac{s+1}{2}k + \frac{2n-s-1}{2}(k-1) > \lambda(G_\text{even}) = \frac{s-1}{2}k + \frac{2n-s+1}{2}(k-1).
        \]
    \end{itemize}
    The proof of~\Cref{lem:gadget-odd-even} follows.
\end{proof}

\begin{lemma}
\label{lem:theta-2-hard}
    \textsc{Gallai Vertex} is $\Theta^p_2$-hard.
\end{lemma}
\begin{proof}
    \begin{figure}[htpb]
        \centering
        \includegraphics[width=0.55\textwidth,page=3]{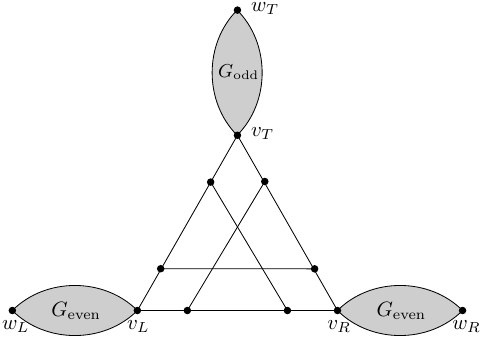}
        \caption{The instance used to show $\Theta^p_2$-hardness.}
        \label{fig:theta-2-hardness}
    \end{figure}
    We reduce from \textsc{SAT Parity}. Given an instance $(\varphi_1,\dots,\varphi_{2n})$ of \textsc{SAT Parity}, we first compute the gadgets $G_\text{even}$ and $G_\text{odd}$ as in \cref{lem:gadget-odd-even}. Then we construct the graph $G$ as in \cref{fig:theta-2-hardness}.
    It is similar to the graph from \cref{fig:NP-coNP-hardness}, except that we now append some gadget to all three vertices $T,L,R$. We use the gadget $G_\text{even}$ twice, and the gadget $G_\text{odd}$ once. We can then consider the cases:
    \begin{itemize}
        \item If $s$ is even, then $\lambda(G_\text{even}) = \lambda(G_\text{odd})$, so in all three endings $T,L,R$, the longest path has the same length.
        Then by \cref{claim:gadgets-instead-of-paths}, the graph $G$ does not have a Gallai vertex. 
      \item If $s$ is odd, then $\lambda(G_\text{even}) < \lambda(G_\text{odd})$. Hence by \cref{claim:gadgets-instead-of-paths}, the graph $G$ has a Gallai vertex.
    \end{itemize}
    Note that for different numbers $s$, the values  $\lambda(G_\text{even}), \lambda(G_\text{odd})$ are different, but this does not influence correctness of the argument. We have shown that $G$ has a Gallai vertex if and only if $s$ is odd. This shows that \textsc{Gallai Vertex} is $\Theta^p_2$-hard, and hence completes the proof of~\Cref{lem:theta-2-hard}.
\end{proof}

\section{Strong Inapproximability}

In this section, we consider approximation algorithms to the longest path transversal problem. 
We assume that such an approximation algorithm runs in polynomial time, and upon receiving a graph $G$ outputs some number $\mathcal{A}(G) \geq \lpt(G)$. 
Its \emph{approximation ratio} is the worst-case value of the ratio $\mathcal{A}(G)/\lpt(G)$ over all graphs $G$.
We show that the longest path transversal number is \emph{strongly inapproximable}. We use the term strongly inapproximable, to denote the following behavior. 

\begin{theorem}
\label{thm:LPT-strong-inapproximable}
For all constants $C \ge 2$, if there exists some graph $H$ with $\lpt(H) = C$, then the computational problem to distinguish whether a given graph $G$ has $\lpt(G) \geq C$ or $\lpt(G) = 1$ is $\Theta^p_2$-hard (and in particular NP-hard).
\end{theorem}

Since Zamfirescu \cite{zamfirescu1976longest} found a graph $G$ with $\lpt(G) = 3$, this implies that $\lpt(G)$ cannot be approximated better than a factor of 3.
The theorem furthermore implies the following: Should the answer to \cref{ques:constantupper} be true, and there exists a (smallest) universal constant $U$ such that $\lpt(G) \leq U$ for all graphs, then the algorithm that always outputs $\mathcal{A}(G) = U$ trivially has an approximation ratio of $U$. 
On the other hand, by our theorem this bound is tight. 
Should the answer to \cref{ques:constantupper} be false, then our theorem implies that there cannot exist a constant-factor approximation algorithm (unless $\text{P} = \text{NP}$).

We are now ready to prove our theorem. We first give a short overview of the idea of the proof: The main idea is 
somewhat analogous to 
our $\Theta^p_2$-hardness proof of \cref{lem:theta-2-hard}. 
Observe that there, we took the Walther-Zamfirescu graph and attached gadgets to the three vertices $L,R,T$. 
These three vertices are exactly the vertices of the Walther-Zamfirescu graph that are endpoints of some longest path. 
Exactly one of the gadgets was a copy of $G_\text{odd}$, and all the remaining gadgets (in this case, two) were copies of $G_\text{even}$.
The new proof works in the same basic way, except that we now swap the Walther-Zamfirescu graph for the graph $H$. 
We do not have any information about the structure of $H$, except that $\lpt(H) = C$, yet, 
this suffices for our argument. We attach a copy of $G_\text{even}$ to every vertex, except to one vertex, which must be the endpoint of some longest path, where we attach a copy of $G_\text{odd}$.
This way, we obtain a modified graph $H'$ from $H$. 
We show that for this new graph $H'$ it is hard to distinguish $\lpt(H') \geq C$ and $\lpt(H') = 1$.

\begin{proof}[Proof of \cref{thm:LPT-strong-inapproximable}]
\begin{figure}
    \centering
    \includegraphics[width=0.8\textwidth]{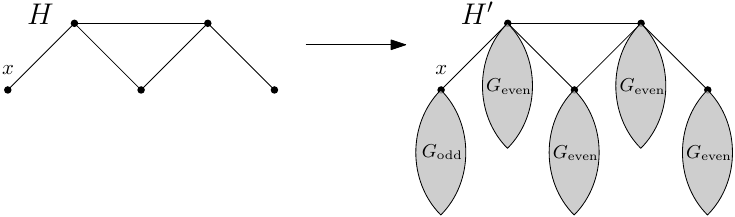}
    \caption{The instance used to show strong inapproximability of $\lpt(G)$.}
    \label{fig:LPT-inapproximability}
\end{figure}
    Assume there exists a constant $C$, and a graph $H$ with $\lpt(H) = C$. 
    Let $x \in V(H)$ be a vertex of $H$ such that at least one longest path of $H$ starts at $x$.
    We want to show the following:
    Given an instance $(\varphi_1,\dots,\varphi_{2n})$ of \textsc{SAT Parity}, we can compute in polynomial time a new graph $H'$ such that if the split index $s$ is even, then $\lpt(H') = C$, and if the split index $s$ is odd, then $\lpt(H') = 1$. This suffices to show the theorem.

    The graph $H'$ is sketched in \cref{fig:LPT-inapproximability} and defined as follows. First, we compute from the sequence $(\varphi_1,\dots,\varphi_{2n})$ in polynomial time gadgets $G_\text{odd}$ and $G_\text{even}$ as in \cref{lem:gadget-odd-even}.
    Then, we append a copy of gadget $G_\text{even}$ to all vertices $v \in V(H) \setminus \set{x}$ (identifying the vertices $v$ and $v_e$). 
    Finally, we append gadget $G_\text{odd}$ to vertex $x$ (identifying vertices $x$ and $v_o$). Let for $v \in V(H)$ denote $G_v$ the even/odd gadget attached to vertex $v$. This completes the description of $H'$. We can now do a case distinction on the split index.
    \begin{itemize}
        \item \textbf{If $s$ is even}: Then $\lambda(G_\text{odd}) = \lambda(G_\text{even})$ by \cref{lem:gadget-odd-even}. We will argue that this yields a one-to-one correspondence between longest paths in $H^\prime$ and longest paths in $H$, in a way that results in a one-to-one correspondence between their longest path traversals.  
        What are the longest paths in $H'$? Since the gadgets are connected to the rest of $H'$ via cut-vertices, and due to the properties of the gadgets, we conclude that 
        \[
        \lambda(H') = \lambda(H) + 2\lambda(G_\text{even}).
        \]
        Furthermore, let $\mathcal{L}(H')$ be the set of all longest paths in $H'$. Then it is straightforward () to check that the following two sets are equal
        \begin{align*}
        \mathcal{L}(H') = \{P_1 P_2 P_3 : \exists a,b \in V(H) : P_2 \text{ is a $a$-$b$-path of length $\lambda(H)$ in $H$,} \\
        \text{$P_1$ is a longest path in $G_a$, $P_3$ is a longest path in $G_b$}\}.
        \end{align*}
        Here $P_1P_2P_3$ denotes the concatenation of paths $P_1,P_2,P_3$. From this characterization of $\mathcal{L}(H')$ it is easy to see that any longest path transversal of $H'$ can be turned into one of $H$ and vice-versa. Hence $\lpt(H') = \lpt(H) = C$.

        \item \textbf{If $s$ is odd}: Then $\lambda(G_\text{odd}) > \lambda(G_\text{even})$ by \cref{lem:gadget-odd-even}. Then, using a similar reasoning, we see that the longest path in $H'$ has length $\lambda(H) + \lambda(G_\text{odd}) + \lambda(G_\text{even})$. Here we use the fact that there is only a single copy of the gadget $G_\text{odd}$ attached to vertex $x$, and that there exists at least one longest path in $H$ starting at $x$.
        Then every longest path in $H'$ uses the vertex $x$, hence $\lpt(H') = 1$.
    \end{itemize}
    We have shown that from a given instance of $\textsc{SAT Parity}$, we can in polynomial time construct a graph $H'$ such that $\lpt(H') = C$ if $s$ is even, and $\lpt(H') = 1$ otherwise. Note that the algorithm runs in polynomial time, since $H,C,x$ are fixed constants in this context and need not be computed. This completes the proof of~\Cref{thm:LPT-strong-inapproximable}.
\end{proof}

As a corollary, we obtain that the longest path transversal number is strongly inapproximable, even when restricted to planar graphs only. 

\begin{corollary}
\label{cor:LPT_planar}
   For all constants $C \ge 2$, if there exists some planar graph $H$ with $\lpt(H) = C$, then the computational problem to distinguish whether a given planar graph $G$ has $\lpt(G) \geq C$ or $\lpt(G) = 1$ is $\Theta^p_2$-hard.
\end{corollary}
\begin{proof}
    It is known that the Hamiltonian path problem is NP-complete for planar graphs \cite{garey1976planar}, even when the start- and endpoint $x_1,x_2$ of the Hamiltonian path are required to lie in the same face $f$. As a general fact about planar graphs, we can w.l.o.g.\ assume that $f$ is the outer face. 
    Due to this, we can assume that the gadget $G(\varphi)$ from \cref{lem:gadget-ham-path} as well as the
    the gadgets $G_\text{odd}$ and $G_\text{even}$ from \cref{lem:gadget-odd-even} are planar, without losing any of the properties described in these lemmas. Then, we can repeat the same proof as in \cref{thm:LPT-strong-inapproximable}. Note that if $H$ is planar, so is $H'$. The rest of the argument is completely analogous to \cref{thm:LPT-strong-inapproximable}.
\end{proof}

Since planar graphs $G$ with $\lpt(G) = 2$ and $\lpt(G) = 3$ were found by Walther, Schmitz, and Zamfirescu (\cite{schmitz1975langste,zamfirescu1975graphen}, see also \cite{shabbir2013intersecting}), this implies immediately:

\begin{corollary}
\label{cor:gallai-vertex-planar}
   \textsc{Gallai Vertex} is $\Theta^p_2$-complete, even when restricted to planar graphs.
\end{corollary}

\section{Longest Cycle Transversal}

In this section, we show that the results from the previous section about the $\Theta^p_2$-completeness and strong inapproximability of the longest path transversal problem carry over with slight adaptations to the longest cycle transversal problem. Let $G$ be a graph. The \emph{longest cycle transversal number} of $G$, denoted by $\lct(G)$, is the minimum size of a longest cycle transversal of $G$, that is, a subset of vertices $S \subseteq V(G)$ such that every longest cycle in $G$ contains at least one vertex from $S$. We denote by $\gamma(G)$ the length of longest cycle in $G$.
In contrast to the situation for the longest path transversal, there are examples of connected $n$-vertex graphs where every longest cycle transversal must have size $\Theta(n)$ (see~\cite{shabbir2013intersecting, zamfirescu2001intersecting}).

\begin{theorem}
\label{thm:LCT-strong-inapproximable}
For all constants $C \ge 2$, the computational problem to distinguish whether a given graph $G$ has $\lct(G) \geq C$ or $\lct(G) = 1$ is $\Theta^p_2$-hard.
\end{theorem}
\begin{proof}
    \begin{figure}[b]
        \centering
        \includegraphics[width=1\textwidth]{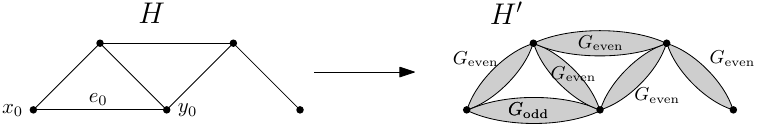}

        \caption{The instance used to show strong inapproximability of $\lct(G)$.}
        \label{fig:lct-strongly-inapproximable}
    \end{figure}
    By \cite{shabbir2013intersecting, zamfirescu2001intersecting}, there exists a graph $H$ with $\lct(H) \ge C$
    . Note that $H$ has a cycle. Let $e_0=x_0y_0$ be some edge of $H$ used by at least one longest cycle.
    Given an instance $(\varphi_1,\dots,\varphi_{2n})$ of $\textsc{SAT Parity}$, we can in polynomial time construct a new graph $H'$ such that $\lct(H') \geq C$ if the split index $s$ is even, and $\lct(H') = 1$ otherwise (see \cref{fig:lct-strongly-inapproximable}).
    We construct $H'$ as follows: From $(\varphi_1,\dots,\varphi_{2n})$, compute the gadgets $G_\text{odd}, G_\text{even}$ in polynomial time. Replace every edge $e \in E(H) \setminus \set{e_0}$ by a copy of the gadget $G_\text{even}$, and the edge $e_0$ by a copy of the gadget $G_\text{odd}$.
    We now consider a case distinction on the split index $s$:
    \begin{itemize}
        \item \textbf{If $s$ is even}: Then $\lambda(G_\text{odd}) = \lambda(G_\text{even})$ by \cref{lem:gadget-odd-even}. 
        Now, cycles of length $L$ in $H$ give rise to cycles of length $L\cdot\lambda(G_{\text{even}})$ in $H'$. What are the longest cycles in $H'$? 
        First observe that if a cycle is contained fully inside some gadget, then it has length at most $\lambda(G_\text{odd})$ or $\lambda(G_\text{even})$, so it cannot be a longest cycle. 
        Any cycle that intersects some gadget and is not contained said gadget must use both vertices $v_e$ and $w_e$ (or both $v_o$ and $w_o$). 
        A longest cycle will use exactly $\lambda(G_\text{even})$ edges from any gadget it visits.
        Hence
        \[
        \gamma(H') = \gamma(H)\lambda(G_\text{even}).
        \]
        Furthermore, let $\mathcal{C}(H')$ be the set of all longest cycles in $H'$. For an edge $e_i\in H$, let $G_{e_i}$ be the (even or odd!) gadget replacing $e_i$. Then, the following two sets are equal:
        \begin{align*}
        \mathcal{C}(H') = \{P_{e_1}P_{e_2}\ldots P_{e_t} \mid \exists e_1,\dots, e_t \in E(H) \text{ forming a longest cycle in $H$,} \\
        \text{$P_{e_i}$ is a longest path in $G_{e_i}$ for all $i=1,\dots,t$}\}.
        \end{align*}
        It follows from this characterization of $\mathcal{C}(H')$ that $\lct(H') = \lct(H) = C$.
        
        \item \textbf{If $s$ is odd}: Then $\lambda(G_\text{odd}) > \lambda(G_\text{even})$ by \cref{lem:gadget-odd-even}. Then, using a similar reasoning, we see that the longest cycle in $H'$ has length \[
        (\gamma(H) - 1)\lambda(G_\text{even}) + \lambda(G_\text{odd}). 
        \] 
        Here we use the fact that there is only a single copy of the gadget $G_\text{odd}$, which replaced the edge $e_0$, and that there exists at least one longest cycle in $H$ using $e_0$. Now, all longest cycles using $e_0$ in $H$ lift to cycles in $H'$ that are just one longer than those cycles that did not use $e_0$.
        Thus, now, every longest cycle in $H'$ uses the vertex $x_0$, hence $\lct(H') = 1$. 
    \end{itemize}
        So, from a given instance of $\textsc{SAT Parity}$, one can in polynomial time construct a graph $H'$ such that $\lct(H') = C$ if $s$ is even, and $\lct(H') = 1$ otherwise. This completes the proof. 
\end{proof}

We have now seen that for general graphs $G$, the longest cycle transversal number $\lct(G)$ can not be approximated to a constant factor in polynomial time. 
However, note that in the above, the graphs may contain \emph{cut vertices}, i.e. vertices whose removal would disconnect the graph.
We say a graph is \textit{$\mathit{2}$-connected} if it contains no cut vertices.

The notion of $2$-connectivity has seen much attention in the study of longest cycle transversals,
as this is a more natural setting: unlike when the graph contains cut vertices, in a $2$-connected graph, any two longest cycles must intersect.
As a final result, we apply our main idea to $2$-connected graphs. We can even further assume the graph to be planar.

\begin{theorem}
\label{thm:LCT-strong-inapproximable-2con}
For all constants $C \ge 2$, if there exists a planar, 2-connected graph $H$ with $\lct(H) = C$, then the computational problem to distinguish if some given planar, 2-connected graph $G$ has $\lct(G) \geq C$ or $\lct(G) = 1$ is $\Theta^p_2$-hard.
\end{theorem}
\begin{proof}
    We adapt the proof of~\cref{thm:LCT-strong-inapproximable}. It is known that the Hamiltonian path problem is NP-complete even for planar, $2$-connected graphs, where the start- and endpoint $x_1,x_2$ share a face \cite{garey1976planar}. 
    Thus, we can construct the gadgets $G_\text{even}$ (and $G_\text{odd}$) using the reduction in \cite{garey1976planar}.   
    Then, observe that the graph $G'_\text{even}$ 
    constructed by taking the  gadget $G_\text{even}$ and 
    connecting $v_e$ and $w_e$ by a dummy edge, 
    is again planar and $2$-connected. (Analogously for $G_\text{odd}$.) 
    Now, consider the process of taking a $2$-connected graph, and substituting edges one by one by either $G_\text{even}$ or $G_\text{odd}$. If such a substitution were to introduce a cut vertex, that cut vertex would also correspond to a cut vertex of $G'_\text{even}$ or $G'_\text{odd}$. Thus, by the end of this process, the resulting graph is again $2$-connected.
    The rest of the proof is analogous to \cref{thm:LCT-strong-inapproximable}.
\end{proof}

Since Zamfirescu \cite{zamfirescu1975graphen} found a 2-connected planar graph with $\lct(H) = 3$ (with 914 vertices) it follows that $\lct(G)$ on 2-connected planar graphs is hard to approximate better than a factor of 3. This also implies our final corollary. (Note that containment in $\Theta^p_2$ is analogous to \cref{lem:containment}.)

\begin{corollary}\label{cor:longest-cycle-singleton-theta2-complete}
The problem of determining whether some graph has a single vertex hitting all longest cycles is $\Theta^p_2$-complete, even for planar and 2-connected input graphs.
\end{corollary}

\section{Maximum Clique and Other Transversal Problems}

In this section, we quickly discuss a general observation. 
We showed in the previous sections that it is $\Theta^p_2$-complete to decide whether there exists a transversal of size 1 for the longest path, or the longest cycle in a graph. 
It is natural to ask whether an analogous result holds for many other settings where a transversal of size 1 is sought after.
In this section, we answer the question positively for the maximum clique problem (and analogously the maximum independent set problem). 
We conjecture that very similar results hold as well for many other problems.
In the following, we use the term \emph{maximum clique} to denote any clique that is cardinality-maximal among all cliques of a given graph. Consider the problem \textsc{Singleton-Max-Clique-Transversal}.

\begin{mybox}{\textsc{Singleton-Max-Clique-Transversal}}
\textbf{Input:} A connected graph $G$.\\
\textbf{Question:} Is there a vertex $v$ such that every maximum clique of $G$ includes $v$?
\end{mybox}

\begin{theorem}
\textsc{Singleton-Max-Clique-Transversal} is $\Theta^p_2$-complete.
\end{theorem}
\begin{proof}
    For the containment in $\Theta^p_2$, we argue analogously to \cref{lem:containment}. Given a graph $G$ on $n$ vertices, we make $O(n^2)$ nonadaptive NP-queries as follows. 
    For each $v \in V$ and each $k \in \{1,\dots,n\}$, query if $G - v$ contains a clique of size at least $k$. For $k \in \{1,\dots,n\}$, query if $G$ contains a clique of size at least $k$.
    This gives enough information to deduce whether there exists a vertex $v$ such that the size of the maximum clique in $G - v$ is smaller than in $G$. Therefore $\textsc{Singleton-Max-Clique-Transversal} \in \Theta^p_2$.

    For the hardness, we reduce from $\textsc{SAT Parity}$. 
    As a first step, let a SAT-formula $\varphi$ be given. 
    Karp \cite{DBLP:conf/coco/Karp72} famously showed that the \textsc{Clique} problem is NP-complete. 
    Karps's result implies that one can in polynomial time compute from $\varphi$ some graph $H(\varphi)$ and integer $k$ such that if $\varphi$ is satisfiable, 
    then $\omega(H(\varphi)) = k$, and if $\varphi$ is not satisfiable then $\omega(H(\varphi)) < k$. Here $\omega(\cdot)$ denotes the size of the largest clique in the graph.
    We can define a new graph $G(\varphi)$ from $H(\varphi)$ by taking the disjoint union with a clique of size $k-1$. Then $\omega(G(\varphi)) = k$ if $\varphi$ is satisfiable and $\omega(G(\varphi)) = k-1$ otherwise.

    Consider now an instance $(\varphi_1,\dots,\varphi_{2n})$ of $\textsc{SAT Parity}$.
    We can due to the above compute gadgets $G(\varphi_1),\dots,G(\varphi_{2n})$ and integers $k_1,\dots,k_{2n}$ such that $\omega(G(\varphi_i))$ equals either $k_i$ or $k_i - 1$, depending on whether $\varphi_i$ is satisfiable. 
    Let $k := \max_{i=1}^{2n} k_i+1$. We can similarly to \cref{lem:gadget-odd-even} assume w.l.o.g. that $k_i = k$ for all $i =1,\dots,2n$. 
    Indeed, for all $i = 1,\dots,2n$, we can modify gadget $G(\varphi_i)$ by introducing a set $W_i$ of $k - k_i$ new vertices connected to all other vertices (including themselves). 
    Note that $W_i \neq \emptyset$ and $W_i$ is contained in every maximal clique of $G(\varphi_i)$.
    
    We define a gadget $G_\text{odd}$ by a two-step process. We first take the disjoint union $G(\varphi_1) \cup G(\varphi_3) \cup \dots \cup G(\varphi_{2n-1})$. 
    Then we add all possible edges between vertices of $G(\varphi_i)$ and $G(\varphi_j)$ for $i \neq j$. This completes the definition of $G_\text{odd}$. 
    Similarly, $G_\text{even}$ is defined as the disjoint union $G(\varphi_2) \cup G(\varphi_4) \cup \dots \cup G(\varphi_{2n})$ and adding all edges between $G(\varphi_i)$ and $G(\varphi_j)$ for $i \neq j$.
    It is clear that both these gadgets are connected graphs, and that $\omega(G_\text{odd})$ and $\omega(G_\text{even})$ depend on the split index $s$ of the sequence $(\varphi_1,\dots,\varphi_{2n})$. Very similar to \cref{lem:gadget-odd-even}, we can argue that if $s$ is even, then 
        \[
            \omega(G_\text{odd}) = \omega(G_\text{even}) =  \frac{s}{2}k + \frac{2n-s}{2}(k-1).
        \]
        \item However, if $s$ is odd, then 
        \[
        \omega(G_\text{odd}) = \frac{s+1}{2}k + \frac{2n-s-1}{2}(k-1) > \omega(G_\text{even}) = \frac{s-1}{2}k + \frac{2n-s+1}{2}(k-1).
        \]

    Finally, for the given instance $(\varphi_1,\dots,\varphi_{2n})$ of $\textsc{SAT Parity}$, the corresponding instance $G'$ of \textsc{Singleton-Max-Clique-Transversal} is given by the disjoint union $G_\text{odd} \cup G_\text{even}$. Note that $G'$ is not yet connected. 
    To fix that, we pick an arbitrary vertex $v$ of $G_\text{odd}$, and add a single edge from $v$ to an arbitrary vertex of $G_\text{even}$. This completes the description of $G'$.
    Note that $G'$ can be computed in polynomial time from the \textsc{SAT Parity} instance. Furthermore, if $s$ is even, then $\omega(G_\text{odd}) = \omega(G_\text{even})$. 
    Hence there are two disjoint cliques of maximum size, so $G'$ is a no-instance of   \textsc{Singleton-Max-Clique-Transversal}.
    If in contrast $s$ is odd, then $\omega(G_\text{odd}) > \omega(G_\text{even})$, so all maximum cliques are contained in $G_\text{odd}$. Furthermore, by our construction every maximum clique contains the vertex set $W_1 \cup W_3 \cup \dots W_{2n-1}$. 
    Hence any vertex from this set hits all maximum cliques. We conclude that $G'$ is a yes-instance of \textsc{Singleton-Max-Clique-Transversal}. This reduction shows that \textsc{Singleton-Max-Clique-Transversal} is $\Theta^p_2$-hard.
\end{proof}

As a final remark, we offer the following technical insight. Consider the only slightly differently defined problem \textsc{Singleton-k-Clique-Transversal}, where an explicit parameter $k$ is part of the input.

\begin{mybox}{\textsc{Singleton-k-Clique-Transversal}}
\textbf{Input:} A connected graph $G$, integer $k$.\\
\textbf{Question:} Is there a vertex $v$ such that every clique of size at least $k$ includes $v$?
\end{mybox}

Then the complexity of this problem behaves differently. It is shown in \cite{grune2026complexity} that the problem $\textsc{Singleton-k-Clique-Transversal}$ is contained in coNP. Since $\text{coNP} \subseteq \Theta^p_2 \subseteq \Sigma^p_2$, this means that \textsc{Singleton-k-Clique-Transversal} is likely not $\Theta^p_2$-complete, unless the polynomial hierarchy collapses.
If one asks instead of a single vertex hitting all cliques for a set of vertices hitting all cliques, then the complexity of the problem changes again, and it becomes $\Sigma^p_2$-complete. Similar observations hold for many other transversal problems. We refer the reader to \cite{grune2026complexity} for a detailed discussion.

\newpage

\bibliography{ref}

\end{document}